\documentclass[twoside]{article}
\usepackage{amsmath,amsfonts}
\usepackage{textcomp}
\input{qic.sty}

\newenvironment{proof}{\par\noindent{\bf Proof.} }{\hfill$\square$\par}

\newtheorem{corollary}[theorem]{Corollary}

\newcommand{\cent}{\mbox{\textcent}}
\newcommand{\mymatrix}[2]{\left( \begin{array}{#1} #2\end{array} \right)}
\newcommand{\myvector}[1]{\mymatrix{c}{#1}}
\newcommand{\myrvector}[1]{\mymatrix{r}{#1}}

\begin{document}
\setlength{\textheight}{8.0truein}

\runninghead{Rational-Valued Affine Verifiers}{Z. Chen and J. Wu}

\normalsize\textlineskip
\thispagestyle{empty}
\setcounter{page}{1}

\copyrightheading{0}{0}{2026}{000--000}

\vspace*{0.88truein}

\alphfootnote

\fpage{1}

\centerline{\bf RATIONAL-VALUED AFFINE VERIFIERS IN ARTHUR--MERLIN PROOF SYSTEMS}
\vspace*{0.035truein}
\centerline{}
\vspace*{0.37truein}

\centerline{\footnotesize ZEYU CHEN\fnm{a}}
\vspace*{0.015truein}
\centerline{\footnotesize\it School of Mathematical Sciences, Zhejiang University}
\baselineskip=10pt
\centerline{\footnotesize\it Hangzhou 310058, People's Republic of China}
\vspace*{10pt}

\centerline{\footnotesize JUNDE WU\fnm{b}}
\vspace*{0.015truein}
\centerline{\footnotesize\it School of Mathematical Sciences, Zhejiang University}
\baselineskip=10pt
\centerline{\footnotesize\it Hangzhou 310058, People's Republic of China}
\fnt{a}{Corresponding author(s). E-mail(s): chenzeyu@zju.edu.cn.}
\fnt{b}{Contributing authors: wjd@zju.edu.cn.}
\vspace*{0.225truein}

\vspace*{0.21truein}

\abstracts{
Affine automata provide a finite-state computational model that preserves the linear-algebraic structure of quantum computation while operating entirely over the reals. Recent work has shown that affine automata can far surpass classical probabilistic finite-state verifiers. However, prior constructions relied on arbitrary real-valued transition matrices, leaving open whether the observed power stems from the affine mechanism itself or from computational resources implicitly encoded in irrational or infinite-precision parameters. This paper studies one-way and two-way automata with deterministic and affine states as verifiers in Arthur--Merlin proof systems under the restriction that every affine transition matrix has rational entries, and shows that the resulting rational model still supports the main verification advantages of affine finite-state verification. At the one-way level, we verify benchmark nonregular languages that are provably hard or impossible for classical two-way probabilistic verifiers. At the two-way level, we achieve weak verification of every Turing-recognizable language, strong bounded-error verification for every language in $\mathbf{ATIME}(2^{O(n)})$, and perfect-completeness strong verification for every language in $\mathbf{PSPACE}$. These results establish that the remarkable verification power of affine finite-state automata is structural.
}{}{}

\vspace*{10pt}

\vspace*{6pt}

\section{Introduction}\label{sec:introduction}

Finite automata provide one of the most fundamental frameworks for understanding computation under constant internal memory. Deterministic and nondeterministic finite automata (DFAs and NFAs), introduced by Rabin and Scott \cite{rabin1959}, capture precisely the regular languages. Enriching the transition structure with linear algebra leads to qualitatively richer behavior: randomization yields probabilistic finite automata (PFAs) \cite{rabin1963}, while quantum interference gives rise to quantum finite automata (QFAs) \cite{kondacs1997,moore2000measure-once,ambainis1998,AmbainisW02,ambainis2006latvian}; see also the surveys \cite{SY14,AY21}. These extensions demonstrate that even constant-memory machines can exhibit surprisingly rich computational phenomena once linear-algebraic structure is introduced.

Affine automata, introduced by D\'{\i}az-Caro and Yakary{\i}lmaz \cite{diaz2016}, isolate a particularly clean form of this richer behavior. Like PFAs, they apply linear updates to a finite register, but they also permit negative coordinates, and acceptance is determined by taking absolute values and normalizing by the \(\ell_1\)-norm. This yields an interference-like effect without complex amplitudes or unitarity, making affine automata a natural bridge between probabilistic and quantum models. In automata theory, they offer a clean setting for studying the computational role of cancellation in signed linear systems; in quantum computing, they serve as a quantum-inspired benchmark that retains interference while dispensing with the Hilbert-space structure of genuine quantum evolutions. A growing body of work shows that affine automata can surpass PFAs and several QFA variants in recognition power, exact computation, succinctness, and verification tasks \cite{villagra2018,nakanishi2017,nakanishi2022exact,ibrahimov2018,ibrahimov2021error,khadieva2021,chen2025two}; see also \cite{hirvensalo2018} for a broader discussion of interference as a computational resource.

Interactive proof systems are a natural arena in which to measure this additional power. In an Arthur--Merlin (AM) proof system \cite{babai1985}, a resource-bounded verifier interacts with an all-powerful prover using public randomness; the space-bounded perspective was further developed in \cite{condon1989complexity,Con93}. For constant-memory classical verifiers, the verification power of two-way probabilistic machines is now well understood through the work of Dwork and Stockmeyer and related earlier results \cite{freivalds1981,DworkStockmeyer1990,dwork1992}. More recently, semi-quantum and affine verifiers have been shown to go significantly beyond this classical finite-state baseline \cite{say2017,ZhengQiuGruska2015,chen2025two}. In particular, two-way affine automata with deterministic and affine states (2ADfAs) can verify every language with bounded error \cite{chen2025two}.

This paper investigates the same question under the additional requirement that all affine transitions have rational entries. This restriction is significant for both technical and conceptual reasons. It keeps the verifier finitely describable, in the same explicit sense as standard rational PFAs \cite{rabin1963}, and it excludes hidden computational power that might arise from arbitrary real constants or infinite precision \cite{say2017}. Any remaining verification advantage must therefore originate from the affine mechanism itself---signed cancellation together with \(\ell_1\)-norm observation---rather than from the choice of transition parameters. Rational-valued affine automata thus provide a particularly rigorous model for isolating the structural power of interference-like computation.

The main contributions of this paper are as follows.
\begin{enumerate}
\item \textbf{One-way affine verifiers.} We construct perfect-completeness AM protocols for the fixed-center middle language and the fixed-center palindrome language, both operating in real time with a single affine register. These two nonregular languages are provably hard or impossible for classical two-way probabilistic verifiers, so these protocols demonstrate a concrete and immediate verification advantage for affine machines.
\item \textbf{Two-way affine verifiers: history-streaming route.} We first establish a weak protocol that verifies every Turing-recognizable language by streaming and checking a computation history. We then introduce a probabilistic continuation check based on \emph{sweep-based metronome clocking} of the input head. This upgrade yields strong verification with bounded error for deterministic computations whose verified histories have polynomial or exponential total length.
\item \textbf{Two-way affine verifiers: game-reduction route.} We develop an independent approach via an alternating subset-sum game language. By streaming a linear-space reduction to this PSPACE-complete problem and verifying the game interactively, we obtain perfect-completeness strong verification for every language in \(\mathbf{PSPACE}\).
\end{enumerate}

Taken together, these results establish that the remarkable verification power of affine finite-state automata is entirely preserved when the transition matrices are restricted to rational entries. The advantage is therefore structural, rooted in the algebraic mechanism of signed cancellation and \(\ell_1\)-norm observation, rather than an artifact of irrational parameters or infinite precision. More broadly, the paper reinforces the role of affine automata as a compelling bridge between automata theory and quantum computing: the model is simple enough to isolate interference-like cancellation as a computational resource, yet expressive enough to support verification phenomena that far exceed the classical finite-state setting.

The remainder of the paper is organized as follows. Section~\ref{sec:affine-computation} reviews affine computation. Section~\ref{sec:automata-and-AM} introduces ADfAs and the AM framework. Section~\ref{sec:1ADfA} presents the one-way protocols. Section~\ref{sec:history-streaming} develops the history-streaming route, beginning with weak verification for Turing-recognizable languages and then strengthening it via continuation checks. Section~\ref{sec:game-reduction} develops the game-reduction route and derives the PSPACE consequence. Section~\ref{sec:conclusion} concludes.

\section{Affine computation}\label{sec:affine-computation}
Inspired by quantum systems, affine systems generalize probabilistic systems by allowing states to take negative values, evolve via linear transformations, and extract information through operations analogous to quantum measurements. In this section, we introduce the fundamental notions of affine systems and discuss specific affine operators. We refer readers to \cite{diaz2016,hirvensalo2018} for further background on affine computation and interference-based viewpoints.

\subsection{Basics of affine systems}
An $m$-state affine register with basis $\{e_{1},\dots,e_{m}\}$ lives in $\mathbb{R}^{m}$. An \emph{affine state} is a vector
\[
v=\begin{pmatrix}v_{1}\\\vdots\\v_{m}\end{pmatrix}\in\mathbb{R}^{m},
\qquad
\sum_{j=1}^{m}v_{j}=1.
\]
We denote
\[
e_j \;=\; \myvector{0 \\ \vdots \\ 0 \\ 1 \\ 0 \\ \vdots \\ 0} 
\quad\leftarrow\ \text{$j$th entry}
\]
when the register is definitely in basis state $e_{j}$. For any affine state $v$, the $i$th entry is denoted $v_i$.

An \emph{affine operator} is an $m\times m$ matrix
\[
A=(a_{ij})\in\mathbb{R}^{m\times m},
\qquad
\sum_{i=1}^{m}a_{ij}=1\ \text{for each column }j,
\]
which maps affine states to affine states via $v' = Av$.

To retrieve information from an affine register, we apply a \emph{weighting operator}. For $v$ as above, after weighting, the probability of observing $e_{j}$ is
\[
P(j)=\frac{|v_{j}|}{\|v\|_{1}} \in[0,1],
\]
where $\|v\|_{1}=\sum_i |v_i|$ is the $\ell_1$-norm. Upon observation $j$, the register collapses to $e_{j}$. Weighting is analogous to projective measurement in quantum systems. However, affine states obey the strict normalization rule that entries sum to $1$, which prevents a register from remaining in ``superposition'' after a partial observation. For example, consider the affine vector
\[
\myrvector{1 \\ -1 \\ 1}.
\]
If we attempt a \emph{partial weighting} that separates the basis into
\(\{e_1,e_2\}\) and \(\{e_3\}\), the unnormalized outcomes are
\[
\myrvector{1 \\ -1 \\ 0 } \quad\text{and}\quad \myvector{0 \\ 0 \\ 1}.
\]
The first vector cannot be a legal affine state, since no rescaling makes its entry-sum equal to~$1$.
Therefore, computations that require partial observation employ multiple affine registers: some registers are weighted while others preserve the ``superposition.''

\subsection{Elementary affine operators and string encoding}\label{sec:elementary-affine-operators}
It follows from the definition that the composition of two affine operators is affine, and the inverse of an invertible affine operator is also affine. 

Suppose we have an affine state 
\[
v=\begin{pmatrix}
x_1 ,\, x_2 ,\,\ldots,\,x_n,\,y,\,\bar{1}
\end{pmatrix}^{T},
\]
where $\bar{1}$ is a balancing entry that preserves the sum-to-$1$ condition throughout this paper. We can compute the linear combination $s=c_1x_1+c_2x_2+\cdots+c_nx_n$ and overwrite $y$ with $s$ by applying the affine operator
\[
\begin{pmatrix}
    x_1\\x_2\\ \vdots \\x_n\\s\\ \bar{1}
\end{pmatrix}
=\begin{pmatrix}
1 &0 &\cdots &0 & 0& 0\\
0 &1 &\cdots &0 & 0& 0\\
\vdots &\vdots &\ddots&\vdots &\vdots& \vdots\\
0 &0 & \cdots&1 & 0& 0\\
c_1 & c_2& \cdots &c_n &0 &0\\
-\,c_1 & -\,c_2 &\cdots &-\,c_n & 1 &1\\
\end{pmatrix}
\begin{pmatrix}
    x_1 \\ x_2 \\\vdots\\x_n\\y\\ \bar{1}
\end{pmatrix}.
\]

Given an ordered alphabet $\mathcal{A}=\{\sigma_0,\sigma_1,\ldots,\sigma_{n-1}\}$, let $\mathcal{A}^\ast$ denote all strings over $\mathcal{A}$, including the empty string $\varepsilon$. Define $\mathrm{val}:\mathcal{A}^\ast \rightarrow \mathbb{N}$ by
\[
\mathrm{val}(w)=\sum_{k=1}^{l} i_k\,n^{\,l-k}
\]
for $w=\sigma_{i_1}\sigma_{i_2}\ldots \sigma_{i_l}$. Thus the current symbol is appended as the least significant new base-$n$ digit when the string is scanned from left to right.
We can encode $val(w)$ using a three-state affine register. Start in the state $\begin{pmatrix} 1,0,0 \end{pmatrix}^T$ and, when the current symbol is $\sigma_{k}$, apply
\begin{equation}
    A_k=\mymatrix{ccc}{
    1 & 0 & 0 \\
    k & n & 0 \\
    -k & 1-n & 1 \\
    }.
\end{equation}
After the whole string is read, the state becomes $\begin{pmatrix} 1,\,\mathrm{val}(w),\,-\mathrm{val}(w) \end{pmatrix}^{T}$, and the value of the string is stored in the second entry.


\subsection{Calculating polynomials and exponents} \label{sec:polynomial-exponent}
Let $p(x)$ be a degree-$d$ polynomial. We read the string $0^l$ and encode $p(l)$ in a designated entry of an affine state.

First, encode $l$ using a three-state register. Start in $\begin{pmatrix}1,0,0\end{pmatrix}^{T}$ and, for each symbol $0$, apply
\begin{equation*}
    \mymatrix{rrr}{
        1 & 0 & 0\\
        1 & 1 & 0\\
        -1 & 0 & 1\\}.
\end{equation*}
If after $i$ symbols the state is
\[
v_i=\myrvector{1\\i\\-i},
\]
then after reading the $(i+1)$-th symbol we have
\[
v_{i+1}=\mymatrix{rrr}{
        1 & 0 & 0\\
        1 & 1 & 0\\
        -1 & 0 & 1\\
    } 
    \myrvector{1\\i\\-i}
=\begin{pmatrix}
    1\\i+1\\-(i+1)
\end{pmatrix}.
\]
By induction, the final state is
\[
v_f=\myrvector{1\\l\\-l},
\]
so $l$ is stored in entry $2$.

Since $(i+1)^2=i^2+2i+1$, we can compute $(i+1)^2$ as a linear combination of $i^2,i,1$ using the method from Section~\ref{sec:elementary-affine-operators}. To encode $l^2$, apply, for each symbol,
\begin{equation*}
    \mymatrix{rrrr}{
        1 & 0 & 0& 0 \\
        1 & 1 & 0& 0 \\
        1 & 2 & 1& 0 \\
        -2 & -2 & 0& 1 \\
    }
\end{equation*}
with initial state $v_0=\begin{pmatrix} 1, 0, 0, 0 \end{pmatrix}^{T}$. The induction step is
\begin{equation*}
    v_{i+1}=\mymatrix{rrrr}{
        1 & 0 & 0& 0 \\
        1 & 1 & 0& 0 \\
        1 & 2 & 1& 0 \\
        -2 & -2 & 0& 1 \\
    }\begin{pmatrix}
          1\\i\\i^2\\-i-i^2
    \end{pmatrix}
=\begin{pmatrix}
        1\\i+1\\1+2i+i^2\\
        -(i+1)-(i+1)^2
    \end{pmatrix}.
\end{equation*}

We generalize to $l^d$ via the binomial expansions
\begin{equation}\label{eq:binomial-expansions}
(i+1)^k=\sum_{j=0}^{k}\begin{pmatrix}
     k\\
     j 
\end{pmatrix} i^j.
\end{equation}
Thus a degree-$d$ polynomial $p(l)$ is a linear combination of $1,l,\ldots,l^d$. Define the affine update for symbol $0$ as the composition of two operators: the first updates the first $(d+1)$ entries using \eqref{eq:binomial-expansions}, and the second computes $p(i+1)$ as the corresponding linear combination:
\[
\begin{pmatrix}
    1\\i\\ \vdots\\i^d\\p(i)\\ \bar{1}
\end{pmatrix}\longrightarrow\begin{pmatrix}
    1\\i+1\\ \vdots\\(i+1)^d\\p(i)\\ \bar{1}
\end{pmatrix}\longrightarrow\begin{pmatrix}
    1\\i+1\\ \vdots\\(i+1)^d\\p(i+1)\\ \bar{1}
\end{pmatrix}.
\]

We can also encode $a^l$ when reading a string of length $l$ for a real number $a$. We use a two-state register, start in $v_0=\begin{pmatrix}1,0\end{pmatrix}$, and for each symbol $0$ apply
\begin{equation*}
    M_a=\mymatrix{cc}{
    a & 0\\
    1-a & 1
    }.
\end{equation*}
The final state is
\[
v_f=\myvector{a^l\\1-a^l},
\]
so $a^l$ is stored in the first entry.

\section{Affine finite automata as verifiers}\label{sec:automata-and-AM}
A model of computation specifies how control, memory, and communication are organized. In this section, we use affine systems to define automata with deterministic and affine states and describe how they act as verifiers in interactive proof systems. Our verifier model builds on the affine-automata and affine-verifier frameworks developed in \cite{diaz2016,khadieva2021,chen2025two}. We assume that readers are familiar with automata theory, especially deterministic, nondeterministic, and alternating Turing machines (abbreviated as DTM, NTM, and ATM, respectively), and their time- and space-bounded complexity classes \(\mathcal{X}\mathrm{TIME}\) and \(\mathcal{X}\mathrm{SPACE}\), where \(\mathcal{X}\) is \(\mathbf{D}\) for deterministic, \(\mathbf{N}\) for nondeterministic, and \(\mathbf{A}\) for alternating computations. We refer the reader to \cite{Sip13} for the basics of Turing machines and automata theory, to \cite{chandra1981alternation,hartmanis1965computational} for the complexity-theoretic resource framework used later, and to \cite{Con93} for an excellent review of space-bounded interactive proof systems.  

\subsection{Automata with deterministic and affine states}\label{sec:automata}
We begin with the (one-way) deterministic finite automaton (DFA), which reads the input from left to right. Formally, a DFA is a $5$-tuple
\[
M=(S,\Sigma,\delta,s_I,S_a),
\]
where
\begin{enumerate}
  \item $S=\{s_1,s_2,\ldots,s_m\}$ is a finite set of states;
  \item $\Sigma$ is the input alphabet and $\tilde{\Sigma}=\Sigma\cup\{\cent,\$\}$ augments it with left and right end-markers;
  \item $\delta:S\times \tilde{\Sigma}\to S$ is the transition function.
  \item $s_I\in S$ is the initial state.
  \item $S_a\subseteq S$ is the set of accepting states.
\end{enumerate}

The automaton works on a semi-infinite tape whose squares are numbered $0,1,2,\ldots$. The input $w\in\Sigma^*$ is padded as $\tilde{w}=\cent w\$$ on a read-only tape. The machine starts in $s_I$ on $\cent$. If $\delta(s,\sigma)=s'$, the automaton enters $s'$ and advances the head one cell to the right. It halts after reading the end-marker $\$$. If the current state is an accepting state, the machine accepts the input. Otherwise, the machine rejects the input.

We extend DFAs by equipping the machine with $k>0$ affine registers (defined in Section~\ref{sec:affine-computation}) that can be updated. A (one-way) automaton with deterministic and affine states (ADfA) is a $6$-tuple
\[
M=(S,\Sigma,\delta,s_I,S_a,\{R_1,\ldots,R_k\}),
\]
where the deterministic part $(S,\Sigma,s_I,S_a)$ is as in a DFA. Each register
\[
R_i=(E_i,\mathcal{A}_i,F_i),\ \ E_i=\{e_{i,1},\ldots,e_{i,m_i}\},\ \ \mathcal{A}_i=\{A_{i,1},\ldots,A_{i,\ell_i}\},\ F_i\subseteq E_i
\]
has basis states $E_i$ (with initial basis element fixed as $e_{i,1}$), a finite set of affine operators $\mathcal{A}_i$, and an accepting set $F_i$. Unless stated otherwise, all affine operators have rational entries. 

The transition function of ADfA is
\[
\delta: S\times\tilde{\Sigma}\to S\times\big(\mathcal{A}_1\times\cdots\times\mathcal{A}_k\big),\qquad
\delta(s,\sigma)=(s',\,O_1,\ldots,O_k),
\]
meaning that, on $(s,\sigma)$, the deterministic state updates to $s'$ and each register $R_i$ is updated by $O_i\in\mathcal{A}_i$.

The ADfA performs exactly one weighting step, after reading the right end-marker $\$$. If the current deterministic state belongs to $S_a$, the verifier weights each affine register once. The input is accepted if and only if all observed outcomes $\tau_i$ lie in their accepting sets, i.e., $e_{i,\tau_i}\in F_i$ for every $i$; otherwise it is rejected. If the current deterministic state is not in $S_a$ when $\$$ is read, the input is rejected.

Permitting the head to stay put or move left yields two-way models. A two-way deterministic finite automaton (2DFA) is a $6$-tuple
\[
M=(S,\Sigma,\delta,s_I,S_a,S_r),
\]
with components as before except that $\delta:S\times \tilde{\Sigma}\to S\times\{-1,0,+1\}$ and $S_r\subseteq S$ is the set of rejecting states. The machine starts in $s_I$ scanning $\cent$, and a move $(s',d)$ updates the state to $s'$ and the head by $d\in\{-1,0,+1\}$. The machine moves the head left for $d=-1$, right for $d=+1$, or keeps it stationary when $d=0$. The tape head is not allowed to leave the string $\tilde{w}=\cent w\$$. Unlike a DFA, halting of 2DFA is controlled by the deterministic states: the computation halts immediately upon entering a state in $S_a$ (accept) or in $S_r$ (reject).

A two-way ADfA (2ADfA) extends a 2DFA with affine registers as in the definition of ADfA. Formally, a 2ADfA is a $7$-tuple
\[
M=(S,\Sigma,\delta,s_I,S_a,S_r,\{R_1,\ldots,R_k\}),
\]
where registers \[
R_i=(E_i,\mathcal{A}_i),\ \ E_i=\{e_{i,1},\ldots,e_{i,m_i}\},\ \ \mathcal{A}_i=\{A_{i,1},\ldots,A_{i,\ell_i}\},
\] and $\delta=(\delta_a,\delta_c)$ with
\begin{align}
\delta_a(s,\sigma) = (O_1,\ldots,O_k),\qquad O_i\in \mathcal{A}_i\cup\{W_i\}, \label{eq:delta-a-2way}\\
\delta_c(s,\sigma,\tau_1,\ldots,\tau_k) = (s',d),\qquad d\in\{-1,0,+1\}. \label{eq:delta-c-2way}
\end{align}
Here $W_i$ denotes the unique weighting operation on register $R_i$. We write $\tau_i=0$ if $O_i\in\mathcal{A}_i$ and $\tau_i\in\{1,\ldots,m_i\}$ if $O_i=W_i$, in which case $\tau_i$ is the observed basis index. Unlike ADfA, 2ADfA may weight many times during the computation. Explicitly, each computation step executes first the affine part and then the deterministic part: given the current deterministic state $s$ and scanned symbol $\sigma\in\tilde{\Sigma}$, the affine phase applies $O_i$ to each register $R_i$, producing outcomes $\tau_i$ only when $W_i$ is used; afterwards, the deterministic phase updates the state and head by $(s',d)=\delta_c(s,\sigma,\tau_1,\ldots,\tau_k)$. The halting condition differs from the one-way case: a 2ADfA halts immediately upon entering a state in $S_a$ (accept) or in $S_r$ (reject), and no final weighting is performed.

Because each weighting operation $W_i$ probabilistically collapses $R_i$ to one of several basis states, the computation of a 2ADfA on input $w$ unfolds as a branching tree. Each node represents a configuration
\[
(s,j,v_1,\ldots,v_k),
\]
where $s\in S$ is the deterministic state, $j$ is the head position, and $v_i$ is the current affine state of register $R_i$. A step in which no register is weighted produces a single child; a step that weights one or more registers branches into several children, one for each combination of outcomes, with the corresponding edge labeled by its occurrence probability. The root configuration is $(s_I,0,e_{1,1},\ldots,e_{k,1})$, and the tree may be infinite. Every leaf is a halting configuration that is either accepting or rejecting. We write $Acc_M(w)$ and $Rej_M(w)$ for the total acceptance and rejection probabilities, respectively. The inequality $0\le Acc_M(w)+Rej_M(w)\le 1$ always holds; any deficit equals the probability of non-halting.

\subsection{Interactive proof systems}\label{sec:AM}
An interactive proof system (IPS) consists of a prover (P) and a verifier (V). The prover is computationally unbounded and untrusted; the verifier is resource-bounded and honest. In this work, the verifier \(V\) is an automaton with deterministic and affine states: either an ADfA or a 2ADfA (Section~\ref{sec:automata}).

The verifier has an extra set of communication states \(S_{\mathrm{com}}\subseteq S\) and a fixed communication alphabet \(\Gamma\). There is a write map
\[
\chi:S_{\mathrm{com}}\to \Gamma,
\]
and a single shared communication cell that is initially blank. Whenever the verifier enters a state \(s\in S_{\mathrm{com}}\), it writes \(\chi(s)\) to the cell; the prover immediately overwrites the cell with a reply symbol \(\rho\in\Gamma\). Then, the verifier changes its deterministic state to $s'$ according to the current deterministic state $s$ and the reply symbol $\rho$. To model streaming messages, we may reserve two symbols \(1\) (the request symbol) and \(\#\) (the end-marker) that do not occur in the input \(w\) and let the prover respond to successive requests \(1\) with symbols of a target string \(x\#\) in order.

Let \(Acc_V(w)\) and \(Rej_V(w)\) be the total probabilities of acceptance and rejection, so \(0\le Acc_V(w)+Rej_V(w)\le 1\). We say \(V\) verifies a language \(L\subseteq\Sigma^*\) with error \(\epsilon<1/2\) if there exists a prover \(P\) such that:
\begin{enumerate}
    \item (Completeness) For all \(w\in L\), the pair \((P,V)\) accepts \(w\) with probability at least \(1-\epsilon\).
    \item (Soundness) For all \(w\notin L\) and all provers \(P^*\), the pair \((P^*,V)\) rejects \(w\) with probability at least \(1-\epsilon\).
\end{enumerate}
Equivalently, soundness may be stated as \(\Pr[(P^*,V)\text{ accepts }w]\le \epsilon\) if non-halting is counted as rejection. We can relax the soundness condition as follows:
\begin{enumerate}
    \item[2'] For all \(w\notin L\) and all \(P^*\), the pair \((P^*,V)\) accepts \(w\) with probability at most \(\epsilon\).
\end{enumerate}
Protocols satisfying conditions (1) and (2) are called strong, while those satisfying (1) and (2') are called weak. For one-way automata, this distinction is moot; since the verifier always halts in real time, the probability of non-halting is zero, making the strong and weak soundness conditions equivalent.

An IPS is Arthur--Merlin (AM) if, at every step, the verifier reveals to the prover its new deterministic state and head move, and any weighting outcomes \(\tau_i\) that occurred. Thus the prover has complete information about \(V\)'s random choices and classical trajectory (public coins). The affine registers themselves remain internal except for disclosed outcomes.

We will use $\mathbf{IP}(\cdot)$ and $\mathbf{AM}(\cdot)$ to represent the complexity classes verifiable by IP and AM systems, respectively. When all transitions are restricted to rational entries, we add a subscript \(q\): \(\mathbf{IP}_q(\cdot)\), \(\mathbf{AM}_{\textnormal{q}}(\cdot)\). A protocol has perfect completeness if condition (1) holds with $\epsilon=0$, i.e., every $w\in L$ is accepted with probability $1$. We write $\mathbf{AM}_{\textnormal{q}}^{=1}(\cdot)$ for the subclass of AM protocols with rational transitions and perfect completeness. In the following, we always assume $\epsilon$ to be rational.

To situate our verifier simulations, we recall the classical relationships between alternating and deterministic complexity, established in the seminal paper of Chandra, Kozen, and Stockmeyer \cite{chandra1981alternation}. Alternation trades off with resources as follows: alternating time is captured by deterministic space with only quadratic overhead in the converse direction, and alternating space collapses to deterministic time with exponential blowup in the space bound.

\begin{theorem}[Chandra--Kozen--Stockmeyer]\label{thm:facts}
Let $t(n)\ge \log n$ be time-constructible and $s(n)\ge \log n$ be space-constructible. Then:
\begin{enumerate}
    \item $\mathbf{ATIME}\!\big(t(n)\big)\ \subseteq\ \mathbf{DSPACE}\!\big(t(n)\big)$.
    \item $\mathbf{DSPACE}\!\big(s(n)\big)\ \subseteq\ \mathbf{ATIME}\!\big(s(n)^2\big)$.
\end{enumerate}
\end{theorem}

\section{Examples of languages verified by ADfAs}\label{sec:1ADfA}
In this section, we demonstrate the verification power of one-way ADfAs in an AM setting through two concrete languages. Fix an alphabet \(\Sigma=\{\sigma_0,\sigma_1,\ldots,\sigma_{n-1}\}\) and a distinguished center symbol \(c\notin\Sigma\). We consider
\[
L_{\mathrm{middle}}=\{xcy\mid x,y\in\Sigma^*,\ |x|=|y|\}
\]
and
\[
L_{\mathrm{mpal}}=\{xcx^R\mid x\in\Sigma^*\}.
\]
The language \(L_{\mathrm{middle}}\) is known to require exponential expected time for AM(2PFA) verifiers \cite{dwork1992,ZhengQiuGruska2015}, and \(L_{\mathrm{mpal}}\) cannot be verified by any AM(2PFA) protocol \cite{ZhengQiuGruska2015}.

\begin{theorem}\label{thm:Lmiddle}
    $L_{\mathrm{middle}}$ can be verified with bounded error $\epsilon < 1/2$ by an ADfA with a single $3$-state affine register. The protocol achieves perfect completeness.
\end{theorem}

\begin{proof}
The verifier uses one $3$-state affine register, initialized to \(v_0=(1,0,0)^T\). While scanning the input from left to right, it asks the prover whether the current symbol is the unique center symbol.
\begin{itemize}
    \item Before the claimed center, the verifier applies
    \[
    A=\begin{pmatrix}
        1&0&0\\
        1&1&0\\
        -1&0&1
    \end{pmatrix}.
    \]
    \item At the claimed center, the verifier checks that the scanned symbol is \(c\). If not, it rejects immediately. If the check succeeds, it switches to the second phase without changing the register.
    \item After the claimed center, the verifier applies
    \[
    A^{-1}=\begin{pmatrix}
        1&0&0\\
        -1&1&0\\
        1&0&1
    \end{pmatrix}
    \]
    to every remaining symbol.
\end{itemize}
If the scanned symbol is \(c\) outside the claimed center, or if the prover marks no position or more than one position as the center, the verifier rejects deterministically. After reading the right end-marker, it applies
\begin{equation}
  M_F= \begin{pmatrix}
    1&1-\delta&1-\delta\\
    0&\delta&0\\
    0&0&\delta
\end{pmatrix},
\end{equation}
where \(\delta=\frac{1-\epsilon}{2\epsilon}\), and then weights the register. The input is accepted exactly when outcome~1 is observed.

\textbf{Completeness.}
If \(w=xcy\in L_{\mathrm{middle}}\) with \(|x|=|y|=k\), the honest prover marks the unique occurrence of \(c\). The operator sequence is
\[
(A^{-1})^k I A^k,
\]
so the final state before \(M_F\) is \(v_f=v_0=(1,0,0)^T\). Hence the weighted outcome is 1 with probability 1.

\textbf{Soundness.}
Let \(w\notin L_{\mathrm{middle}}\), let \(\ell:=|w|\), and let the prover mark position \(j\). If the symbol at position \(j\) is not \(c\), the verifier rejects immediately. Otherwise the final state before \(M_F\) is
\[
v_f=(A^{-1})^{\ell-j}A^{j-1}v_0=A^m v_0,
\qquad m:=2j-\ell-1.
\]
Since \(w\notin L_{\mathrm{middle}}\), the marked position cannot split the input into two parts of equal length, so \(m\neq 0\). A routine induction gives \(A^m v_0=(1,m,-m)^T\) for every integer \(m\). Therefore
\[
M_Fv_f=
\begin{pmatrix}
1\\ m\delta\\ -m\delta
\end{pmatrix},
\]
and the acceptance probability is
\[
\Pr[\mathrm{acc}]=\frac{1}{1+2|m|\delta}\leq \frac{1}{1+2\delta}=\epsilon.
\]
Thus the soundness error is at most \(\epsilon\).
\end{proof}

\begin{theorem}\label{thm:lmpal}
    $L_{\mathrm{mpal}}$ can be verified with bounded error $\epsilon$ by an ADfA with $4$ affine states. Moreover, the protocol achieves perfect completeness.
\end{theorem}

\begin{proof}
If $|\Sigma|=1$, then $L_{\mathrm{mpal}}=L_{\mathrm{middle}}$, so the claim follows from Theorem~\ref{thm:Lmiddle}. Assume therefore that \(\Sigma=\{\sigma_0,\sigma_1,\ldots,\sigma_{n-1}\}\) with \(n\ge 2\).

The verifier uses one $4$-state affine register with basis states \(e_1,e_2,e_3,e_4\), where \(e_1\) is the unique accepting basis state and \(e_4\) is the balancing state. The register starts in \(v_0=(1,0,0,0)^T\).

For each \(k\in\{0,1,\ldots,n-1\}\), define
\[
P_k=
\begin{pmatrix}
 n & 0 & 0 & 0\\
 k & 1 & 0 & 0\\
 -k & 0 & 1 & 0\\
 1-n & 0 & 0 & 1
\end{pmatrix},
\qquad
P_k^{-1}=
\begin{pmatrix}
 1/n & 0 & 0 & 0\\
 -k/n & 1 & 0 & 0\\
 k/n & 0 & 1 & 0\\
 1-1/n & 0 & 0 & 1
\end{pmatrix}.
\]
The verifier asks the prover to identify the unique center position. Before the claimed center it applies \(P_k\) whenever the scanned symbol is \(\sigma_k\). At the claimed center it checks that the symbol is \(c\); otherwise it rejects immediately. After the claimed center it applies \(P_k^{-1}\) to each scanned \(\sigma_k\). If the scanned symbol is \(c\) outside the claimed center, or if the prover marks no position or more than one position as the center, the verifier rejects deterministically. At the end it applies
\[
M_F=
\begin{pmatrix}
1&1&1&1\\
0&\delta&0&0\\
0&0&0&\delta\\
0&-\delta&0&-\delta
\end{pmatrix},
\qquad
\delta=\frac{2(1-\epsilon)}{\epsilon},
\]
weights the resulting state, and accepts if and only if outcome~1 is observed.

\textbf{Completeness.}
Let \(w=xcx^R\) with \(x=\sigma_{i_1}\cdots\sigma_{i_s}\). Since the right half is scanned as \(\sigma_{i_s}\cdots\sigma_{i_1}\), the total operator applied before \(M_F\) is
\[
P_{i_1}^{-1}\cdots P_{i_s}^{-1}P_{i_s}\cdots P_{i_1}=I.
\]
Hence the final state before weighting is again \(v_0\), and the verifier accepts with probability~1.

\textbf{Soundness.}
Assume that exactly one center is claimed, and write the scanned input as \(w=xcy\), where
\[
x=\sigma_{i_1}\cdots\sigma_{i_s},
\qquad
y=\sigma_{j_1}\cdots\sigma_{j_t}.
\]
A straightforward induction on the length of the scanned prefix shows that after reading \(x\) the affine state is
\[
\begin{pmatrix}
 n^s\\
 \mathrm{val}(x^R)\\
 -\mathrm{val}(x^R)\\
 1-n^s
\end{pmatrix}.
\]
Applying \(P_{j_1}^{-1},P_{j_2}^{-1},\ldots,P_{j_t}^{-1}\) to the right half yields, again by induction,
\[
v_f=
\begin{pmatrix}
 \rho\\
 \Delta\\
 -\Delta\\
 1-\rho
\end{pmatrix},
\qquad
\rho=n^{s-t},
\qquad
\Delta=\mathrm{val}(x^R)-\rho\,\mathrm{val}(y).
\]
If \(s=t\), then \(\rho=1\) and
\[
\Delta=\mathrm{val}(x^R)-\mathrm{val}(y).
\]
Since \(x\neq y^R\) and \(\mathrm{val}\) is injective on \(\Sigma^s\), we have \(\Delta\neq 0\), so \(|\Delta|\ge 1\). If \(s\neq t\), then \(\rho=n^{s-t}\neq 1\). Because \(n\ge 2\), either \(\rho\ge 2\) or \(\rho\le 1/2\), and therefore
\[
|1-\rho|\ge \frac{1}{2}.
\]
After applying \(M_F\), the affine state becomes
\[
M_Fv_f=
\begin{pmatrix}
1\\
\delta\Delta\\
\delta(1-\rho)\\
-\delta(\Delta+1-\rho)
\end{pmatrix}.
\]
Hence
\[
\Pr[\mathrm{acc}]
=
\frac{1}{1+\delta|\Delta|+\delta|1-\rho|+\delta|\Delta+1-\rho|}.
\]
If \(s=t\), then \(\rho=1\) and \(|\Delta|\ge 1\), so
\[
\Pr[\mathrm{acc}]\le \frac{1}{1+2\delta}<\epsilon.
\]
If \(s\neq t\), then \(|1-\rho|\ge 1/2\), so
\[
\Pr[\mathrm{acc}]\le \frac{1}{1+\delta/2}=\epsilon.
\]
Thus the soundness error is at most \(\epsilon\).
\end{proof}

The two AM protocols above demonstrate a clear and concrete verification advantage for affine verifiers. A real-time ADfA with a single register verifies the nonregular languages \(L_{\mathrm{middle}}\) and \(L_{\mathrm{mpal}}\) with perfect completeness and tunable soundness, whereas AM(2PFA) requires exponential expected time for \(L_{\mathrm{middle}}\) and cannot verify \(L_{\mathrm{mpal}}\) at all \cite{dwork1992,ZhengQiuGruska2015}. On the quantum side, related protocols are known for 2QCFA verifiers, typically with larger expected running times \cite{AmbainisW02,ZhengQiuGruska2015}. This places one-way affine verifiers in a particularly economical position among all known finite-state AM verifiers.

\section{History-streaming verification protocols}\label{sec:history-streaming}
The first general two-way route in this paper is based on streaming a computation history and checking it online with affine registers. This route begins with a weak protocol for Turing-recognizable languages and then strengthens to strong deterministic protocols via continuation checks.

Throughout this section we adopt one common convention for the communicated configurations. Without loss of generality, the simulated single-tape machines work on a fixed marked tape segment whose leftmost symbol is \(\cent\) and whose rightmost symbol is \(\$\). The tape head never moves beyond these two markers, the end-markers are never overwritten, and every internal cell is overwritten only by an internal tape symbol. Accordingly, every communicated configuration has the form \(uqv\), where \(q\) is the current machine state, the head scans the first symbol of \(v\), and the concatenation \(uv\) is the entire current contents of that fixed marked tape segment, from \(\cent\) to \(\$\).

This convention causes no loss of generality. Every finite accepting or rejecting branch of an ordinary single-tape machine visits only finitely many tape cells, so by adding enough initially blank cells and a right end-marker, that branch can be viewed as a computation over such a fixed segment. Conversely, any computation that stays within a fixed marked segment is also a valid computation on the usual semi-infinite blank tape. In the space-bounded results below, the length of this marked tape string may be taken to be \(O(s(|w|))\).

\subsection{Weak verification by streamed computation histories}\label{sec:weak-protocol}
Let \(\mathcal{D}=(Q,\Gamma_{\mathrm{tape}},\delta,q_0,q_{\mathrm{acc}},q_{\mathrm{rej}})\) be a deterministic single-tape Turing machine recognizing \(L\). The tape alphabet \(\Gamma_{\mathrm{tape}}\) contains \(\Sigma\cup\{\cent,\$\}\), and the transition function has the form
\[
\delta:Q\times\Gamma_{\mathrm{tape}}\to Q\times\Gamma_{\mathrm{tape}}\times\{-1,0,+1\}.
\]
We choose the verifier--prover communication alphabet \(\Gamma\) to contain symbols encoding \(\Gamma_{\mathrm{tape}}\cup Q\) together with the separator symbol \(\#\).
\begin{theorem}\label{thm:wk}
    Any Turing-recognizable language can be weakly verified by a rational-valued 2ADfA with arbitrary bounded error $\epsilon$. Moreover, the protocol achieves perfect completeness.
\end{theorem}

\begin{proof}
The prover is expected to stream a sequence
\begin{equation}\label{eq:trans}
c_0\#c_1\#c_2\#\cdots
\end{equation}
of configurations separated by \(\#\), where all blocks are written with respect to the same fixed marked work segment, and the first block is the initial configuration
\[
c_0=q_0\cent w\$,
\]
and each \(c_{i+1}\) is the successor of \(c_i\). Any malformed block, invalid separator pattern, or locally inconsistent successor step causes immediate rejection whenever the verifier can detect it. The only remaining source of non-halting is that a dishonest prover may keep streaming symbols forever without completing a halting computation.

The verifier uses two $4$-state affine registers. One register, denoted \(R_{\mathrm{cmp}}\), is used to compare the received configuration with the successor predicted from the previous block. The other, denoted \(R_{\mathrm{next}}\), simultaneously builds the predicted successor of the current block for the next comparison. Both registers start in \((1,0,0,0)^T\).

To append a symbol \(j\) to entry~2 or entry~3 we use the affine operators
\[
A_j=\begin{pmatrix}
1&0&0&0\\ j&n&0&0\\ 0&0&1&0\\ -j&1-n&0&1
\end{pmatrix},
\qquad
B_j=\begin{pmatrix}
1&0&0&0\\ 0&1&0&0\\ j&0&n&0\\ -j&0&1-n&1
\end{pmatrix},
\]
where \(n=|\Gamma|\). Fix an ordering \(\Gamma=\{\gamma_0,\ldots,\gamma_{n-1}\}\) and use the base-\(n\) encoding \(\mathrm{val}(\cdot)\) from Section~\ref{sec:elementary-affine-operators}. Since all streamed configuration blocks have the same length, equality of blocks is equivalent to equality of their \(\mathrm{val}\)-encodings.

\textbf{Protocol.}
\begin{enumerate}
\item While the prover streams the first block \(c_0\), the verifier compares it symbol by symbol with the pattern \(q_0\cent w\$\). At the same time it computes \(\mathbf{next}(c_0)\) in entry~2 of \(R_{\mathrm{next}}\).
\item After the separator following \(c_0\) is read, the verifier rejects if the initial-block test failed. Otherwise it swaps the roles of the two registers, so that entry~2 of \(R_{\mathrm{cmp}}\) now stores \(\mathrm{val}(\mathbf{next}(c_0))\), and reinitializes the other register for the next successor computation.
\item When the prover streams a later block \(c_i\) with \(i\ge 1\), the verifier appends its symbols into entry~3 of \(R_{\mathrm{cmp}}\) by the operators \(B_j\).
\item At the same time, it computes $\mathbf{next}(c_i)$ in entry~2 of $R_{\mathrm{next}}$.
This is done online by maintaining a finite-state buffer, which can be implemented by encoding information in finite memory. Specifically, any valid configuration can be parsed as $c_i = u' r q a v'$, where $q \in Q$ is the state, $a \in \Gamma_{\mathrm{tape}}$ is the scanned symbol, and $r$ is the tape symbol immediately preceding $q$ (or the left end-marker $\cent$). The verifier processes the prefix $u'$ normally, appending each symbol $\gamma$ into entry~2 by applying $A_\gamma$. When the verifier buffers the local boundary $r q a$, it computes the transition $\delta(q, a) = (q', b, d)$ with $d \in \{-1, 0, +1\}$ and then appends the updated local block in the correct order by applying the corresponding sequence of affine operators (note that matrices are applied to the state vector from right to left):
\[
\begin{array}{ll}
d=-1: & A_b A_r A_{q'},\\
d=0:  & A_b A_{q'} A_r,\\
d=+1: & A_{q'} A_b A_r.
\end{array}
\]
After this boundary update is complete, the verifier empties the buffer and resumes appending the remaining symbols of $v'$ sequentially by applying $A_\gamma$ for each symbol until the separator $\#$ is reached.
\item When the separator \(\#\) is reached, the verifier compares the second and third entries of \(R_{\mathrm{cmp}}\). It first amplifies both entries by
\[
T_C=\begin{pmatrix}
1&0&0&0\\
0&C&0&0\\
0&0&C&0\\
0&1-C&1-C&1
\end{pmatrix},
\qquad C=\frac{1-\epsilon}{2\epsilon},
\]
and then subtracts them by
\[
S=\begin{pmatrix}
1&0&0&0\\
0&1&-1&0\\
0&-1&1&0\\
0&1&1&1
\end{pmatrix}.
\]
Thus the comparison register becomes
\[
\begin{pmatrix}
1\\ C\bigl(\mathrm{val}(\mathbf{next}(c_{i-1}))-\mathrm{val}(c_i)\bigr)\\ C\bigl(\mathrm{val}(c_i)-\mathrm{val}(\mathbf{next}(c_{i-1}))\bigr)\\ 0
\end{pmatrix}.
\]
Weighting this register rejects immediately on outcomes 2 or 3 and passes on outcome 1.
\item If \(c_i\) is halting, the verifier accepts or rejects according to whether \(c_i\) is accepting or rejecting. Otherwise, it swaps the roles of \(R_{\mathrm{cmp}}\) and \(R_{\mathrm{next}}\), reinitializes the new \(R_{\mathrm{next}}\), and proceeds to the next block.
\end{enumerate}

\textbf{Completeness.}
If \(w\in L\), choose an accepting computation branch of \(\mathcal{D}\) and fix a marked work segment satisfying the convention above for that branch. Then the honest prover can stream the corresponding accepting history over this fixed segment. The first block is checked exactly, every successor comparison is exact, and the verifier accepts when the first accepting configuration is reached. Hence the protocol has perfect completeness.

\textbf{Soundness.}
If \(w\notin L\), then no accepting computation history of \(\mathcal{D}\) exists on the usual semi-infinite blank tape, and therefore none exists in the fixed-segment form above. If the prover ever sends a malformed block or a pair of consecutive configurations that is not a legal successor pair, then
\[
\Delta:=\mathrm{val}(\mathbf{next}(c_{i-1}))-\mathrm{val}(c_i)
\]
is a nonzero integer. So the probability of incorrectly passing that comparison is at most
\[
\frac{1}{1+2C|\Delta|}\le \frac{1}{1+2C}=\epsilon.
\]
If instead the prover streams only valid configurations, then the verifier rejects the input deterministically or never halts.
\end{proof}

\subsection{Strong verification via continuation checks}\label{sec:strong-protocol}
The weak protocol above elegantly captures the streamed computation history with perfect completeness, but it leaves one loophole: a dishonest prover may keep transmitting symbols forever. To obtain \emph{strong} verification within the same history-streaming route, the verifier must force eventual halting. 

A particularly simple special case occurs for linear-space deterministic computations: once the input has been scanned, the verifier can reuse its input head as a linear counter and force each streamed configuration to have the expected \(O(|w|)\) size. This yields the following consequence.

\begin{corollary}\label{thm:DTM-direct}
If a language \(L\) is recognized by a single-tape deterministic Turing machine running in linear space and time \(O(t(|w|))\), then \(L\) has a perfect-completeness strong AM protocol with a rational-valued 2ADfA verifier and expected running time \(O(|w|t(|w|))\). Consequently,
\[
\mathbf{DSPACE}(n)\subseteq \mathbf{AM}_{\textnormal{q}}^{=1}(\mathrm{2ADfA}).
\]
\end{corollary}
The next theorem handles general space and time bounds by replacing the exact linear-size test with a small rejection experiment that is performed periodically while the prover is streaming symbols.

\begin{theorem}\label{thm:DTM-pr}
Let $L$ be recognized by a single-tape DTM in space $O(s(|w|))$ and time $t(|w|)$. Assume either $s(|w|)t(|w|)=O(|w|^k)$ or $s(|w|)t(|w|)=2^{O(|w|)}$ for some constants. Then $L$ can be strongly verified by a rational-valued 2ADfA with bounded error $\epsilon\in(0,1/2)$ and expected running time $O\bigl(|w|\,s(|w|)t(|w|)/\epsilon\bigr)$.
\end{theorem}

\begin{proof}
Take the weak verifier from Theorem~\ref{thm:wk}. The successor-comparison part of the protocol is unchanged. We only add a separate continuation-check register. Let
\[
N:=|\tilde w|=|w|+2.
\]
The key implementation device is a \emph{sweep-based metronome clocking}: after the initial input scan, the verifier reuses its read-only input head as a periodic clock for the streamed transcript. Each time it requests one more prover symbol, it also advances the input head by one step to the right; whenever the head reaches the right end-marker, it performs one continuation check and then deterministically sweeps back to the left end-marker to start the next batch. Thus exactly $N$ streamed symbols occur between consecutive continuation checks, and this scheduling is implementable by finite control.

\textbf{Modular protocol execution.}
\begin{itemize}
\item \textbf{Phase 1: Streaming simulation.} The verifier runs the computation-history check from Theorem~\ref{thm:wk} without modification.
\item \textbf{Phase 2: Continuation check.} In parallel with the streamed simulation, the verifier uses the sweep-based metronome clocking to trigger a small rejection experiment once every batch of $N$ prover symbols.
\end{itemize}

\textbf{Implementing the continuation check.}

\emph{Polynomial case.} Assume \(s(|w|)t(|w|)\le cN^k\). Set
\[
p=\frac{1}{mN^{k-1}},\qquad m:=\left\lceil\frac{c}{\epsilon}\right\rceil.
\]
Use a \((k+2)\)-state register initialized to \((1,0,\ldots,0)^T\). During one full scan of \(\tilde w\), apply the operators from Section~\ref{sec:polynomial-exponent} so that the register stores
\[
\myvector{1\\ N-1\\ \vdots\\ (N-1)^{k-1}\\ 0\\ \bar 1}.
\]
At the right end-marker, apply a single affine linear-combination gadget from Section~\ref{sec:elementary-affine-operators} so that the final \(\ell_1\)-norm becomes \(mN^{k-1}\) while the first entry remains 1. Weighting then rejects with probability exactly \(p\).

\emph{Exponential case.} Assume \(s(|w|)t(|w|)\le c2^{kN}\). Set
\[
p=\frac{1}{m2^{kN}},\qquad m:=\left\lceil\frac{c}{\epsilon}\right\rceil.
\]
Use a two-state register, initialized to \((1,0)^T\), and apply
\[
\begin{pmatrix}2^{-k}&0\\ 1-2^{-k}&1\end{pmatrix}
\]
once per scanned input symbol. After the full input has been scanned, the first entry becomes \(2^{-kN}\). Composing once more with
\[
\begin{pmatrix}1/m&0\\ 1-1/m&1\end{pmatrix}
\]
changes the first entry to \(1/(m2^{kN})\), so weighting rejects with probability \(p\).

\textbf{Protocol.}
The verifier from Theorem~\ref{thm:wk} is run exactly as before, except that after every batch of $N$ streamed symbols it performs the continuation check above. If the check rejects, the whole protocol rejects immediately.

\textbf{Completeness.}
Assume \(w\in L\), and let the honest prover stream the unique accepting computation history of the simulated DTM. The successor checks of Theorem~\ref{thm:wk} are always passed. The only possible error is a false rejection by the continuation checks.

If \(s(|w|)t(|w|)\le cN^k\), then the full transcript has length \(O(N^k)\), so there are at most \(cN^{k-1}\) continuation checks. By the union bound,
\[
P_{\mathrm{false\,rej}}
\le cN^{k-1}\cdot \frac{1}{mN^{k-1}}
\le \epsilon.
\]
If \(s(|w|)t(|w|)\le c2^{kN}\), then there are at most \(c2^{kN}/N\) checks, and therefore
\[
P_{\mathrm{false\,rej}}
\le \frac{c2^{kN}}{N}\cdot \frac{1}{m2^{kN}}
\le \frac{c}{m}
\le \epsilon.
\]
Thus the verifier accepts every member with probability at least \(1-\epsilon\).

\textbf{Soundness.}
The continuation checks only add extra rejecting branches. Therefore every accepting computation path of the strong verifier is also an accepting computation path of the weak verifier from Theorem~\ref{thm:wk}. Consequently the acceptance probability on non-members remains at most \(\epsilon\). In addition, if a dishonest prover keeps streaming symbols forever instead of completing a block sequence, the continuation checks form an infinite sequence of independent rejection opportunities, so the probability of non-halting becomes 0.

\textbf{Expected running time.}
For honest inputs, the verifier processes a transcript of total length \(O(s(|w|)t(|w|))\). The extra metronome sweeps contribute only a constant-factor overhead per batch, so the deterministic work done on a correct transcript is still \(O(s(|w|)t(|w|))\). On dishonest inputs that keep streaming forever, the number of continuation checks before rejection is geometric with mean \(1/p\), and one batch contains $N=O(|w|)$ streamed symbols. Hence the extra expected length until rejection is \(O(N/p)\), which is
\[
O\bigl(mN^k\bigr)=O\bigl(|w|s(|w|)t(|w|)/\epsilon\bigr)
\]
in the polynomial case and
\[
O\bigl(mN2^{kN}\bigr)=O\bigl(|w|s(|w|)t(|w|)/\epsilon\bigr)
\]
in the exponential case. This proves the stated bound.
\end{proof}

Combining the above theorem and Theorem~\ref{thm:facts}, we can obtain the following result:
\begin{corollary}
        \textnormal{$\mathbf{ATIME}(2^{O(n)})=\mathbf{DSPACE}(2^{O(n)})\subseteq \mathbf{AM}_{\textnormal{q}}(\text{2ADfA})$}.
\end{corollary}

\section{Game-reduction verification via alternating subset-sum}\label{sec:game-reduction}
The second general two-way route in this paper is reduction-based. Instead of checking a full accepting computation directly, the verifier checks a streamed reduction to a complete interactive problem and consumes the produced instance online. The target problem is an alternating subset-sum game \cite{FearnleyJurdzinski2015}, whose affine suitability comes from the fact that its winning condition is an exact zero-residual test for a linear combination of integers.

\subsection{Verifying the alternating subset-sum game}\label{sec:ssg-route}
An instance consists of a target integer \(S\) and pairs \((a_i,b_i)\) and \((e_i,f_i)\) for \(i=1,\ldots,n\), all written in binary. The universal player first chooses one number from each pair \((a_i,b_i)\); after seeing those choices, the existential player chooses one number from each pair \((e_i,f_i)\). The instance is winning for the existential player if the final sum can always be made equal to \(S\). Formally,
\[
L_{SSG}=\Bigl\{S\,\forall(a_1,b_1)\,\exists(e_1,f_1)\cdots \forall(a_n,b_n)\,\exists(e_n,f_n)\Bigr\},
\]
where membership means that for every choice of \(x_i\in\{a_i,b_i\}\) there exist choices \(y_i\in\{e_i,f_i\}\) such that
\[
S=\sum_{i=1}^n(x_i+y_i).
\]
This game is PSPACE-complete \cite{FearnleyJurdzinski2015}.

\begin{theorem}\label{thm:LSSG}
$L_{SSG}$ can be verified by a 2ADfA with bounded error $\epsilon\in(0,1/2)$. Moreover, the protocol achieves perfect completeness.
\end{theorem}

\begin{proof}
The verifier uses two affine registers.
\begin{itemize}
\item A $3$-state working register stores the residual
\[
R=S-\sum_{i=1}^n(x_i+y_i)
\]
in the form \((1,R,-R)^T\).
\item A $2$-state restart-on-accept register controls a small acceptance probability at the end of each round.
\end{itemize}
Both registers are reinitialized at the beginning of every round.

\textbf{Protocol in one round.}
\begin{enumerate}
\item The verifier reads \(S\) and initializes the working register to \((1,S,-S)^T\).
\item For each \(i=1,\ldots,n\):
\begin{enumerate}
\item The verifier updates the restart register by
\[
\begin{pmatrix}1/2&0\\ 1/2&1\end{pmatrix},
\]
so after all \(n\) universal choices its first entry is \(2^{-n}\).
\item It chooses \(x_i\in\{a_i,b_i\}\) uniformly at random and subtracts \(x_i\) from the residual stored in the working register.
\item The prover responds with some \(y_i\in\{e_i,f_i\}\), which is checked by finite control and then subtracted from the working register.
\end{enumerate}
\item At the end of the round, the working register has the form \((1,R,-R)^T\). The verifier weights it and proceeds only if outcome~1 is observed. Since \(R\) is an integer, this happens with probability
\[
\frac{1}{1+2|R|},
\]
which is \(1\) when \(R=0\) and at most \(1/3\) when \(R\neq 0\).
\item If outcome~1 occurs, the verifier applies
\[
\begin{pmatrix}2\epsilon/3&0\\ 1-2\epsilon/3&1\end{pmatrix}
\]
to the restart register and weights it. It accepts iff the first basis state is observed; otherwise it restarts from the beginning of a fresh round.
\end{enumerate}

\textbf{Completeness.}
If the instance belongs to \(L_{SSG}\), then for every universal sequence \((x_1,\ldots,x_n)\) the prover has a legal response making \(R=0\). Hence every round reaches step~4, and the per-round acceptance probability is exactly
\[
\alpha:=\frac{2\epsilon}{3}\,2^{-n}>0.
\]
There are no rejecting rounds. Therefore the verifier accepts with probability \(1\) over repeated independent rounds.

\textbf{Soundness.}
If the instance does not belong to \(L_{SSG}\), there exists a universal choice sequence \(x^\star\) such that every legal prover response yields \(R\neq 0\). The verifier chooses \(x^\star\) with probability \(2^{-n}\). Conditioned on that event, step~3 rejects with probability at least \(2/3\). Therefore the per-round rejection probability satisfies
\[
r\ge \frac{2}{3}\,2^{-n}.
\]
A round can accept only after step~3 passes and the restart register produces its accepting outcome. Hence the per-round acceptance probability is at most
\[
a\le \alpha=\frac{2\epsilon}{3}\,2^{-n}.
\]
Consequently
\[
\frac{a}{r}\le \epsilon.
\]
Since each round either accepts, rejects, or restarts, the eventual acceptance probability over all restarted rounds is at most
\[
\frac{a}{a+r}\le \frac{a}{r}\le \epsilon.
\]
\end{proof}

\subsection{Linear-space reductions and the PSPACE consequence}\label{sec:pspace-route}
With a verifier for \(L_{SSG}\) in hand, the game-reduction route extends to every language that admits a linear-space reduction to this PSPACE-complete problem. The verifier checks the streamed reduction online and feeds each emitted symbol directly into the subset-sum game protocol, so the reduction output never needs to be stored in full.

\begin{theorem}\label{thm:reduceLSSG}
Any language that is reducible in linear space to $L_{SSG}$ can be strongly verified by a rational-valued 2ADfA with bounded error $\epsilon$. Moreover, the protocol achieves perfect completeness.
\end{theorem}

\begin{proof}
Let \(M\) be a deterministic linear-space transducer that maps each input \(w\) to an instance \(u=M(w)\) of \(L_{SSG}\). We assume that \(M\) has a read-only input tape, a single linear-space work tape, and a write-only output tape. During each transition it may either emit one output symbol or emit nothing. The prover streams an accepting computation history of \(M\), and the verifier processes that history online.

Fix the soundness parameters of the two components as follows: the reduction checker from Corollary~\ref{thm:DTM-direct} is instantiated with error parameter \(\epsilon/2\), while the subset-sum game verifier from Theorem~\ref{thm:LSSG} is instantiated with error parameter \(\epsilon\).

One round of the verifier combines two tasks.
\begin{enumerate}
\item It checks the streamed computation history of \(M\) using the perfect-completeness linear-space protocol from Corollary~\ref{thm:DTM-direct}. Any invalid history causes immediate rejection.
\item Whenever the simulated transducer emits an output symbol, the verifier feeds that symbol to one round of the subset-sum game verifier from Theorem~\ref{thm:LSSG}. Thus the produced instance is never stored in full; it is generated and consumed online.
\end{enumerate}
If the reduction check rejects, the whole round rejects. If the reduction check accepts and the subset-sum game round accepts, the whole protocol accepts. If the reduction check accepts and the subset-sum game round restarts, then the whole protocol restarts from a fresh round, and the prover must stream the reduction computation again.

\textbf{Completeness.}
If \(w\in L\), the honest prover streams the correct accepting computation of \(M\), so the reduction checker never rejects, and the emitted instance satisfies \(M(w)\in L_{SSG}\). The subset-sum game component therefore has perfect completeness by Theorem~\ref{thm:LSSG}. Hence the combined protocol has perfect completeness as well.

\textbf{Soundness.}
If \(w\notin L\), consider any prover strategy.
\begin{itemize}
\item If the streamed reduction history is invalid, then in each round the reduction checker accepts with probability at most \(\epsilon/2\) and rejects with probability at least \(1/2\). Since acceptance of the whole round is possible only after the reduction checker accepts, the eventual acceptance probability over all restarted rounds is at most
\[
\frac{\epsilon/2}{\epsilon/2+1/2} < \epsilon.
\]
\item If the streamed reduction history is valid, then it produces the correct instance \(M(w)\notin L_{SSG}\). In that case the first component has perfect completeness, and the second component is exactly the verifier of Theorem~\ref{thm:LSSG}, whose eventual acceptance probability is at most \(\epsilon\).
\end{itemize}
Thus the whole protocol has soundness error at most \(\epsilon\).
\end{proof}

Since \(L_{SSG}\) is PSPACE-complete and every language in \(\mathbf{PSPACE}\) is log-space reducible to it \cite{FearnleyJurdzinski2015,Yakaryilmaz2025QIP}, Theorem~\ref{thm:reduceLSSG} implies the following corollary.

\begin{corollary}
    \textnormal{$\mathbf{PSPACE}\subseteq \mathbf{AM}_{\textnormal{q}}^{=1}(\mathrm{2ADfA})$.}
\end{corollary}

This completes the game-reduction route. In contrast with the history-streaming constructions of Section~\ref{sec:history-streaming}, it restores perfect completeness and reaches \(\mathbf{PSPACE}\) by replacing direct transcript checking with the streamed verification of a complete interactive game.

\section{Conclusion}\label{sec:conclusion}
This paper studied affine automata as verifiers in Arthur--Merlin proof systems under the requirement that all affine transitions have rational entries. On the one-way side, we constructed real-time perfect-completeness protocols for a fixed-center middle language and a fixed-center palindrome language. Since the first of these languages requires exponential expected time for any AM(2PFA) verifier and the second cannot be verified by any AM(2PFA) protocol at all, these results demonstrate that affine interference yields a significant and immediate verification advantage even in the simplest head-movement regime.

For two-way verifiers, we developed two complementary strategies. The first is a history-streaming route: a weak protocol verifies every Turing-recognizable language by encoding and checking a computation history with perfect completeness, and a probabilistic continuation check based on sweep-based metronome clocking upgrades this to strong verification with bounded error for deterministic computations whose verified histories have polynomial or exponential total length. The second is a game-reduction route: by streaming a linear-space reduction to a PSPACE-complete alternating subset-sum game and verifying the game interactively, we achieved perfect-completeness strong verification for every language in \(\mathbf{PSPACE}\).

Two broad themes emerge from these constructions. First, the present results show that rational-valued affine dynamics retain the central verification strength of earlier affine and quantum finite state models. This indicates that the source of that strength is structural. Second, affine verifiers naturally support two complementary proof paradigms---direct algebraic history checking and restartable reductions to complete interactive problems---and each paradigm offers distinct trade-offs between completeness, running time, and the class of languages captured.

Several natural questions remain open. Can rational-valued 2ADfA verifiers achieve strong verification with perfect completeness for all Turing-recognizable languages? More broadly, it would be valuable to understand how the real-time verification power of rational affine verifiers compares with that of quantum or hybrid verifiers in interactive settings, and whether the algebraic techniques developed here can be extended to other finite-state models.

\nonumsection{Acknowledgements}
The authors thank Professor A. Yakary{\i}lmaz for valuable suggestions and remarks. This project is supported by the National Natural Science Foundation of China (Grants No. 12031004, No. 12271474).

\nonumsection{Author contributions} 
Z.C. was responsible for the conceptualization, methodology, and original draft preparation. J.W. provided the key algorithmic support and facilitated the collaboration. The manuscript was revised collaboratively by Z.C. and J.W. All authors have read and agreed to the published version of the manuscript.

\nonumsection{Data Availability}
No datasets were generated or analysed during the current study.

\nonumsection{Declarations}

\nonumsection{Competing interests}
The authors declare that they have no competing interests.

\bibliographystyle{unsrt}
\bibliography{ref}

\end{document}